\documentclass[envcountsame, envcountsect]{svmult}

\usepackage{enumitem}
\usepackage{bbm}
\usepackage{nicefrac}

\usepackage{type1cm}        % activate if the above 3 fonts are
\usepackage{makeidx}         % allows index generation
\usepackage{graphicx}        % standard LaTeX graphics tool
\usepackage{multicol}        % used for the two-column index
\usepackage[bottom]{footmisc}% places footnotes at page bottom

\usepackage{newtxtext}       %
\usepackage[varvw]{newtxmath}       % selects Times Roman as basic font

\newcommand{\C}{{\mathbb C}}

\newcommand{\h}{\mathbf{h}}

\begin{document}

\title*{Spectral transition model with the general contact interaction}
\author{Pavel Exner and Ji\v{r}\'{\i} Lipovsk\'{y}}

\institute{Pavel Exner \at {Doppler Institute for Mathematical Physics and Applied Mathematics, Czech Technical University,
B\v rehov{\'a} 7, 11519 Prague, Czechia} {\rm and} {Department of Theoretical Physics, Nuclear Physics Institute, Czech Academy of Sciences, 25068 \v{R}e\v{z} near Prague, Czechia}, \email{exner@ujf.cas.cz}
\and Ji\v{r}\'{\i} Lipovsk\'{y} \at Department of Physics, Faculty of Science, University of Hradec Kr\'alov\'e, Rokitansk\'eho 62,
500\,03 Hradec Kr\'alov\'e, Czechia, \email{jiri.lipovsky@uhk.cz}}

\motto{
 Dedicated to the memory of Sergei Naboko
}

%%%%%%%%%%%%%%%%%%%%%%%%NUMBERING%%%%%%%%%%%%%%%%%%%%%%%%
\numberwithin{equation}{section}
\renewcommand{\theequation}{\thesection.\arabic{equation}}
%%%%%%%%%%%%%%%%%%%%%%%%%%%%%%%%%%%%%%%%%%%%%%%%%%%%%%%

\maketitle

\abstract{
Using a technique introduced by Sergei Naboko, we analyze a generalization of the parameter-controlled model of spectral transition, originally proposed by Smilansky and Solomyak, to the situation where the singular interaction responsible for the effect is characterized by the `full' family of four real numbers with the `diagonal' part of the coupling being position-dependent.\\
\\
\textbf{Keywords} Smilansky-Solomyak model $\cdot$ irreversible quantum graph $\cdot$ four-parameter interaction\\
\\
\textbf{Mathematics Subject Classification}\; 81Q35 $\cdot$ 47A10 $\cdot$ 81Q10
}

%%%%%%%%%%%%%%%%%%%%%%
\section{Introduction}
%%%%%%%%%%%%%%%%%%%%%%

Many memories appear when thinking of Sergei Naboko, his charming personality and sharp mathematical mind. Here we choose as a starting point one of his papers \cite{NS06} written with Michael Solomyak, in which he used his deep knowledge of Jacobi operators to enrich our understanding of a model of spectral transition originally proposed by Smilansky and Solomyak \cite{Smi04, SmS05}. This model has different physical interpretations, either as a model of an irreversible behavior on a graph coupled to a caricature heat bath \cite{Smi04} or as the two-dimensional Schr\"odinger operator with a singular interaction of a position-dependent strength \cite{Sol04a, Sol04b}. It contains a parameter that controls its spectrum; it has a critical value at which the spectral nature abruptly changes from a below bounded and partly discrete to an absolutely continuous one covering the whole real axis.

The original model was modified in various ways. For instance, the harmonic confining potential in \eqref{Hamiltonian} below can be replaced by a more general function \cite{Sol06a} and the line on which the coupling is imposed can be replaced by a more general graph \cite{Sol06b}. Another modification consists of replacing the singular interaction by a regular potential channel \cite{BE14, BE17}; the advantage is that one is able in this setting to compare the quantum dynamics with its classical counterpart which exhibits an interesting irregular scattering behavior \cite{Gu18}. We note also that even the basic model still poses open questions, for example, having in addition to the discrete spectrum in the subcritical case also an infinite family of resonances the behavior of which is not fully understood \cite{ELT17}.

A generalization going, so to say, in the opposite direction to \cite{BE14, BE17} consists of replacing the $\delta$ interaction of the original model by a more singular coupling. In \cite{EL18} we did that with the interaction commonly known as $\delta'$ \cite{AGHH}. The aim of the present paper is to extend the conclusions to the situation where such a singular interaction is of the most general type depending on four real parameters; the mentioned $\delta$ and $\delta'$ coupling are now included as particular cases. We are going to show that the effect of abrupt spectral transition is robust, however, it occurs now on a hypersurface in the parameter space. The key element in our analysis of the spectral transition is the link to spectral properties of a specific Jacobi operator introduced in \cite{NS06}.

The paper is structured as follows. In the next section, we introduce our 4-parameter model. Section~\ref{sec:quadratic} is devoted to the study of the quadratic form associated with the Hamiltonian; we prove a lower bound for it. In Section~\ref{sec:jacobi}, we obtain the recurrence system which defines the Jacobi operator associated to the problem. Section~\ref{sec:sa} includes the proof of self-adjointness of the Hamiltonian. In Section~\ref{sec:ac}, the absolutely continuous spectrum of the Hamiltonian is obtained, using the known properties of another Jacobi operator, which differs from the Jacobi operator associated with the system by a compact operator. Theorem~\ref{thm:acj0} similarly to the previous results on $\delta$ and $\delta'$-coupling shows the abrupt change of the spectra of the Jacobi operator from purely absolutely continuous to discrete depending on a real parameter depending on the coupling parameters. Using this result, Theorem~\ref{thm:ach} studies the absolutely continuous spectrum of the Hamiltonian. Finally, in Section 7 we obtain results on the discrete spectrum, in particular, we compare the number of eigenvalues of the Hamiltonian and the Jacobi operator (see Theorem~\ref{thm:discrete2}) and find the asymptotic formul\ae\ for the number of eigenvalues (Theorem~\ref{thm:discrete3}). The Appendix~\ref{sec:appa} includes some technical results needed in Section~\ref{sec:sa}.

%%%%%%%%%%%%%%%%%%%
\section{The model}\label{sec:model}
%%%%%%%%%%%%%%%%%%%

The model we are going to discuss describes a quantum system the Hamiltonian of which is the operator of the form
 % ------------- %
\begin{subequations}
\label{model}
\begin{equation}
  \mathbf{H}_{\alpha,\beta,\gamma} \Psi(x,y) = -\frac{\partial^2 \Psi}{\partial x^2}(x,y) +\frac{1}{2} \left(-\frac{\partial^2 \Psi}{\partial y^2}(x,y)+ y^2\Psi(x,y)\right) \label{Hamiltonian}
\end{equation}
 % ------------- %
with the general contact interaction with position-dependent coefficients supported by the axis $x=0$, characterized by choosing the operator domain as the family of functions in the Sobolev space $\Psi \in H^2((0,\infty)\times \mathbb{R})\oplus H^2((-\infty,0)\times \mathbb{R})$ satisfying the boundary conditions
 % ------------- %
\begin{align}
\hspace{-1em}
\frac{\partial \Psi}{\partial x}(0+,y)-\frac{\partial \Psi}{\partial x}(0-,y) & =  \frac{\alpha}{2}y \big(\Psi(0+,y)+\Psi(0-,y)\big)+\frac{\gamma}{2}\Big(\frac{\partial \Psi}{\partial x}(0+,y)+\frac{\partial \Psi}{\partial x}(0-,y)\Big),\label{eq:coupl1} \\[.3em]
\Psi(0+,y)-\Psi(0-,y) & = -\frac{\bar\gamma}{2} \big(\Psi(0+,y)+\Psi(0-,y)\big)+\frac{\beta}{2y}\Big(\frac{\partial \Psi}{\partial x}(0+,y)+\frac{\partial \Psi}{\partial x}(0-,y)\Big)\label{eq:coupl2}
\end{align}
\end{subequations}
 % ------------- %
with the parameters $\alpha, \beta \in\mathbb{R}$ and $\gamma\in\mathbb{C}$. We note that the conditions defining the general contact interaction are written in different ways \cite[Appendix~K.1]{AGHH}; here we choose the form proposed in \cite{EG99} which has the advantage that one easily singles out the particular cases of $\delta$- and $\delta'$-interactions: the choice $\beta=\gamma =0$ leads to the original Smilansky-Solomyak model with the $\delta$-interaction on the $x$ axis \cite{Smi04, Sol04a, Sol04b, SmS05} and $\alpha=\gamma=0$ yields its $\delta'$-modification discussed in \cite{EL18}. One of the features of those models was that the spectrum was independent of the coupling constant sign. In the general case the mirror transformation, $y\to -y$ can be compensated by a \emph{simultaneous} change of the `diagonal' parameters, $\alpha\to -\alpha$ and $\beta\to -\beta$ which thus leaves the spectrum invariant.

%%%%%%%%%%%%%%%%%%%%%%%%%%%%
\section{The quadratic form}\label{sec:quadratic}
%%%%%%%%%%%%%%%%%%%%%%%%%%%%

A convenient way to deal with the operator \eqref{model} is to use the quadratic form method.

 % ------------- %
\begin{theorem}
The operator $\mathbf{H}_{\alpha,\beta,\gamma}$ is associated with the quadratic form given by
 % ------------- %
\begin{align*}
  \mathbf{a}_{\alpha,\beta,\gamma}[\Psi] &= \mathbf{a}_0[\Psi]+ \frac{1}{\beta} \big(\mathbf{b}_1[\Psi] + \mathbf{b}_2[\Psi]+ \mathbf{b}_3[\Psi]\big)\quad \mathrm{for}\quad \beta \ne 0, \\[.3em]
  \mathbf{a}_{\alpha,0,\gamma}[\Psi] &= \mathbf{a}_0[\Psi]+ \mathbf{b}_4[\Psi]\quad \mathrm{for}\quad \beta = 0,
\end{align*}
 % ------------- %
where
 % ------------- %
\begin{align*}
  \mathbf{a}_0[\Psi] &= \int_{\mathbb{R}^2}\Big(\left|\frac{\partial \Phi}{\partial x}\right|^2+\frac{1}{2}\left|\frac{\partial \Phi}{\partial y}\right|^2+\frac{1}{2}y^2|\Psi|^2\Big)\,\mathrm{d}x\mathrm{d}y, \\[.15em]
  \mathbf{b}_1[\Psi] &= \int_\mathbb{R} y \left|\Psi(0+,y)-\Psi(0-,y)\right|^2\,\mathrm{d}y, \\[.15em]
  \mathbf{b}_2[\Psi] &= \int_\mathbb{R} \frac{\alpha \beta+|\gamma|^2}{4}\,y\left|\Psi(0+,y)+\Psi(0-,y)\right|^2\,\mathrm{d}y, \\[.15em]
  \mathbf{b}_3[\Psi] &= \int_\mathbb{R} y \,\mathrm{Re}\big[\gamma(\bar \Psi(0+,y)+\bar \Psi(0-,y))(\Psi(0+,y)-\Psi(0-,y))\big]\,\mathrm{d}y, \\[.15em]
  \mathbf{b}_4[\Psi] &= \int_\mathbb{R} \frac{\alpha}{4}\,y\left|\Psi(0+,y)+\Psi(0-,y)\right|^2\,\mathrm{d}y
\end{align*}
 % ------------- %
and the domain $D = \mathrm{dom\,}\mathbf{a}_0$ of the form $\mathbf{a}_0$ is
 % ------------- %
$$
  D = \left\{\Psi\in H^1((0,\infty)\times\mathbb{R})\oplus H^1((-\infty, 0)\times\mathbb{R}):\: \mathbf{a}_0[\Psi]<\infty\right\}.
$$
 % ------------- %
\end{theorem}
 % ------------- %
\begin{proof}
The quadratic form is given by the integral over $\mathbb{R}^2$ of the expression $\bar\Psi(x,y) (H\Psi)(x,y)$ where $H$ is the symbol given by the right-hand side of \eqref{Hamiltonian}. By integration by parts in both variables $x,y$ and using the fact that $\bar\Psi(x,y)\frac{\partial \Psi}{\partial x}(x,y)$ vanishes as $x,y \to \pm \infty$ we obtain
 % ------------- %
$$
  \mathbf{a}_{\alpha,\beta,\gamma}[\Psi] = \mathbf{a}_0[\Psi]+\int_\mathbb{R} \Big(\bar\Psi(0+,y)\frac{\partial \Psi}{\partial x}(0+,y)-\bar\Psi(0-,y)\frac{\partial \Psi}{\partial x}(0-,y)\Big)\,\mathrm{d}y.
$$
 % ------------- %
Introducing the shorthands
 % ------------- %
\begin{align*}
  f_+ := \Psi(0+,y)+\Psi(0-,y)\,,&\quad f_- := \Psi(0+,y)-\Psi(0-,y), \\[.3em]
  f_+' := \Psi'(0+,y)+\Psi'(0-,y)\,,&\quad f_-' := \Psi'(0+,y)-\Psi'(0-,y)
\end{align*}
 % ------------- %
one can rewrite the last integral as
 % ------------- %
$$
  \frac{1}{2}\int_\mathbb{R}(\bar f_+ f_-'+\bar f_- f_+')\,\mathrm{d}y.
$$
 % ------------- %

We start with the case $\beta \ne 0$. Rewrite the matching conditions \eqref{eq:coupl1} and \eqref{eq:coupl2} as
 % ------------- %
\begin{align*}
 f_+' &= \frac{2y}{\beta} (f_- + \frac{\bar\gamma}{2}f_+), \\[.15em]
 f_-' &= \frac{\gamma y}{\beta} f_-+\frac{y}{2}(\alpha+\frac{|\gamma|^2}{\beta})f_+\,.
\end{align*}
 % ------------- %
and substituting it into the above equation we arrive at
 % ------------- %
\begin{align*}
  \mathbf{a}_{\alpha,\beta,\gamma}[\Psi] - \mathbf{a}_0[\Psi] &= \int_\mathbb{R}\frac{y}{4\beta}\big[(\alpha\beta+|\gamma|^2)|f_+|^2+4|f_-|^2+4\mathrm{Re\,}(\gamma \bar f_+ f_-)\big]\,\mathrm{d}y \\[.2em]
  &= \frac{1}{\beta}(\mathbf{b}_1[\Psi]+\mathbf{b}_2[\Psi]+\mathbf{b}_3[\Psi]).
\end{align*}
 % ------------- %

For $\beta = 0$ the matching conditions \eqref{eq:coupl1} and \eqref{eq:coupl2} yield
 % ------------- %
$$
 f_- = -\frac{\bar \gamma}{2}f_+\,,\quad f_-'= \frac{\alpha y}{2} f_+ +\frac{\gamma}{2}f_+',
$$
 % ------------- %
so that
 % ------------- %
$$
  \frac{1}{2}(\bar f_+ f_-'+\bar f_- f_+') = \frac{1}{2}\left(\frac{\alpha y}{2}|f_+|^2+\frac{\gamma}{2}\bar f_+ f_+'+\bar f_-f_+'\right) = \frac{\alpha y}{4} |f_+|^2,
$$
 % ------------- %
and as a result, we get
 % ------------- %
$
  \mathbf{a}_{\alpha,0,\gamma}[\Psi] - \mathbf{a}_0[\Psi] = \mathbf{b}_4[\Psi]
$
 % ------------- %
which completes the proof.
\end{proof}
 % ------------- %

Next we state a simple lemma; if a proof is needed it can be found, for instance in \cite{EL18} as Lemma~5.

 % ------------- %
\begin{lemma}\label{lem:tech1}
For all $c, d\in \mathbb{C}$ we have
 % ------------- %
$
  2|\mathrm{Re\,}(\bar c d)|\leq |c|^2+|d|^2\,.
$
 % ------------- %
\end{lemma}
 % ------------- %

The second simple lemma generalizes Lemma~\ref{lem:tech1}:
 % ------------- %
\begin{lemma}\label{lem:tech3}
Let $\sigma_j$, $j=1,2,3$ be the Pauli matrices and $\sigma_0$ the $2\times 2$ identity matrix. Let further $\mathbf{u}, \mathbf{v}$ be complex two-component column vectors. Then for any $\omega_j\in \mathbb{R}$ we have
 % ------------- %
$$
  \mathrm{Re\,}\bigg[\bar{\mathbf{u}}^T \cdot \bigg(\sum_{j=0}^3 \omega_j \sigma_j\bigg)\cdot \mathbf{v}\bigg]\leq \frac{1}{2} \big(\|\mathbf{u}\|^2+\|\mathbf{v}\|^2\big)\Big(|\omega_0|+\sqrt{\omega_1^2+\omega_2^2+\omega_3^2}\Big)
$$
 % ------------- %
\end{lemma}
 % ------------- %
\begin{proof}
Since the real part of a complex number is bounded by its modulus, we have
 % ------------- %
$$
  \mathrm{Re\,}\bigg[\bar{\mathbf{u}}^T \cdot \bigg(\sum_{j=0}^3 \omega_j \sigma_j\bigg)\cdot \mathbf{v}\bigg]\leq
\bigg\|\bar{\mathbf{u}}^T \cdot \bigg(\sum_{j=0}^3 \omega_j \sigma_j\bigg)\cdot \mathbf{v}\bigg\| = \|\mathbf{u}\|\|\mathbf{v}\|\bigg\|\sum_{j=0}^3 \omega_j \sigma_j\bigg\|,
$$
 % ------------- %
where the matrix norm is the operator norm. Using $2\|\mathbf{u}\|\|\mathbf{v}\|\leq \|\mathbf{u}\|^2+ \|\mathbf{v}\|^2$ we complete the proof noting that the eigenvalues of the matrix
 % ------------- %
$$
  \sum_{j=0}^3 \omega_j \sigma_j = \begin{pmatrix}\omega_0 + \omega_3 & \omega_1 - i \omega_2\\ \omega_1 + i \omega_2 & \omega_0 - \omega_3\end{pmatrix}
$$
 % ------------- %
are $\omega_0 \pm \sqrt{\omega_1^2+\omega_2^2+\omega_3^2}$.
\end{proof}
 % ------------- %

For the following particular choice of parameters $\omega_j$, $j=0,\dots, 3$
 % ------------- %
\begin{equation}
  \omega_0 = 4+\alpha\beta+|\gamma|^2\,,\quad \omega_1 = \alpha\beta+|\gamma|^2 -4\,, \quad \omega_2 = 4\,\mathrm{Im\,}\gamma\,,\quad \omega_3 = 4\,\mathrm{Re\,}\gamma\label{eq:omega}
\end{equation}
 % ------------- %
we denote the matrix $\sum_{j=0}^3 \omega_j \sigma_j$ by $\Sigma$. This matrix appears later in the text.

Let the subspaces of $D_j$, $j=1,2$ consist of functions $\psi^{(j)}\in H^1((-\infty,0))\oplus H^1((0,\infty))$ with $\begin{pmatrix}\psi^{(j)}(0+)\\\psi^{(j)}(0-)\end{pmatrix} = \mathrm{Span\,}\begin{pmatrix}K_+^{(j)}\\K_-^{(j)}\end{pmatrix}$, where $\begin{pmatrix}K_+^{(j)}\\K_-^{(j)}\end{pmatrix}$ are eigenvectors of $\Sigma$. Simple calculation yields up to a constant
 % ------------- %
\begin{subequations}\label{eq:kplusminus_subspaces}
\begin{align}
\hspace{-1em}
  K_+^{(1)} &= \omega_3-\sqrt{\omega_1^2+\omega_2^2+\omega_3^3} = 4\mathrm{Re\,}\gamma-\sqrt{(\alpha\beta+|\gamma|^2-4)^2+16|\gamma|^2}\,,\\
  K_-^{(1)} &= \omega_1+ i \omega_2 = \alpha\beta+|\gamma|^2-4+4 i \mathrm{Im\,}\gamma\,,\\
  K_+^{(2)} &= \omega_3+\sqrt{\omega_1^2+\omega_2^2+\omega_3^3} = 4\mathrm{Re\,}\gamma+\sqrt{(\alpha\beta+|\gamma|^2-4)^2+16|\gamma|^2}\,,\\
  K_-^{(2)} &= \omega_1+ i \omega_2 = \alpha\beta+|\gamma|^2-4+4 i \mathrm{Im\,}\gamma\,.
\end{align}
\end{subequations}
 % ------------- %

The following lemma generalizes \cite[Lemma 6]{EL18}
\begin{lemma}\label{lem:tech2}
For all $\delta>0$ and $\psi\in H^1((-\infty,0))\oplus H^1((0,\infty))$ it holds
 % ------------- %
$$
\delta|\psi(0\pm)|^2 \leq \pm\int_{0\pm}^{\pm\infty} \left(|\psi'(x)|^2+\delta^2|\psi(x)|^2\right)\,\mathrm{d}x\,.
$$
 % ------------- %
On the subspaces of $D_j$ these inequalities are saturated for
 % ------------- %
\begin{equation}
  \tilde \psi^{(j)}_\delta(x) := \left\{\begin{matrix}\frac{K_-^{(j)} \mathrm{e}^{\delta x}}{\sqrt{\delta (|K_+^{(j)}|^2+|K_-^{(j)}|^2)}} & \mathrm{for}\quad  x<0\\\frac{K_+^{(j)} \mathrm{e}^{-\delta x}}{\sqrt{\delta (|K_+^{(j)}|^2+|K_-^{(j)}|^2)}}  & \mathrm{for}\quad x>0\end{matrix}\right.\,, \label{eq:psitilde}
\end{equation}
 % ------------- %
with $j=1,2$ and the constants $K_\pm^{(j)}$ given by~\eqref{eq:kplusminus_subspaces}.
\end{lemma}
\begin{proof}
The proof is a minor modification of the one in \cite[Lemma 6]{EL18}. The result follows from the positivity of the integrals $\int_{-\infty}^0 |\psi'(x)-\delta \psi(x)|^2\,\mathrm{d}x$ and $\int_{0}^\infty |\psi'(x)+\delta \psi(x)|^2\,\mathrm{d}x$. The second part of the lemma can be verified by a direct inspection.
\end{proof}
 % ------------- %

With these preliminaries we can derive upper and lower bounds to the forms involved:
 % ------------- %
\begin{theorem}\label{thm:bounds}
We have the following inequalities:
 % ------------- %
\begin{align*}
  \mathbf{b}_1[\Psi] + \mathbf{b}_2[\Psi] + \mathbf{b}_3[\Psi] &\leq   \frac{|4+\alpha\beta+|\gamma|^2|+\sqrt{(\alpha\beta+|\gamma|^2-4)^2+16|\gamma|^2}}{2\sqrt{2}\beta}\, \mathbf{a}_0[\Psi], \\[.15em]
  \mathbf{b}_4[\Psi] &\leq \frac{\alpha}{\sqrt{2}} \mathbf{a}_0[\Psi].
\end{align*}
For $\beta\ne 0$
\begin{align*}
  \mathbf{a}_{\alpha,\beta,\gamma}[\Psi] &\geq \bigg(1-\frac{|4+\alpha\beta+|\gamma|^2|+\sqrt{(\alpha\beta+|\gamma|^2-4)^2+16|\gamma|^2}}{2\sqrt{2}\beta}\,\bigg)\mathbf{a}_0[\Psi]\\[.15em]
 &\geq \frac{1}{2}\left(1-\frac{|4+\alpha\beta+|\gamma|^2|+\sqrt{(\alpha\beta+|\gamma|^2-4)^2+16|\gamma|^2}}{2\sqrt{2}\beta}\right)\|\Psi\|^2\,.
\end{align*}
For $\beta = 0$
\begin{align*}
    \mathbf{a}_{\alpha,0,\gamma}[\Psi] &\geq \left(1-\frac{\alpha}{\sqrt{2}}\right)\mathbf{a}_0[\Psi] \geq \frac{1}{2}\left(1-\frac{\alpha}{\sqrt{2}}\right)\|\Psi\|^2\,.
\end{align*}
 % ------------- %
\end{theorem}
 % ------------- %
\begin{proof}
We proceed a way similar to \cite{EL18}, where the bound for $\mathbf{b}_1[\Psi]$ was obtained, and to \cite{Sol04b}, where bounds for $\mathbf{b}_2$ and $\mathbf{b}_4$ were found, cf. \cite[Lemma 2.1]{Sol04b} or \cite[Theorem~4]{EL18} for more details. To get the first inequality, one has to estimate the expression $\mathbf{b}_1[\Psi]+\mathbf{b}_2[\Psi]+\mathbf{b}_3[\Psi]$. In view of the mentioned results, we can skip a part of the computation related to the first two terms and focus on the form $\mathbf{b}_3[\Psi]$ only.

As in \cite{EL18,Sol04b}, we use the expansion of the function in terms of the `transverse' basis,
 % ------------- %
\begin{equation}
  \Psi(x,y) = \sum_{n\in \mathbb{N}_0}\psi_n(x)\chi_n(y),\label{eq:separation}
\end{equation}
 % ------------- %
with the coefficients $\psi_n$ depending on the variable $x$ only, where $\chi_n(y)$ is the $n$-th Hermite function, that is, the normalized harmonic oscillator eigenfunction referring to the eigenvalue $n+\frac12$, and $\mathbb{N}_0$ denotes the set of non-negative integers. We note that the Ansatz~\eqref{eq:separation} can be also used to analyze the model numerically as it was done in \cite{ELT17} for the eigenvalues and resonances appearing in the subcritical case.

The quadratic form $\mathbf{a}_0$ can be in terms of the coefficient functions $\psi_n$ written as
 % ------------- %
\begin{equation}
  \mathbf{a}_0[\Psi] = \sum_{n\in\mathbb{N}_0}\int_\mathbb{R}\Big(|\psi_n'(x)|^2+\big(n+\textstyle{\frac{1}{2}}\big)|\psi_n(x)|^2\Big)\,\mathrm{d}x,\label{eq:a0}
\end{equation}
 % ------------- %
giving the bound
 % ------------- %
$$
  \mathbf{a}_0[\Psi]  \geq \frac{1}{2}\sum_{n\in\mathbb{N}_0}\int_\mathbb{R}|\psi_n(x)|^2\,\mathrm{d}x = \frac{1}{2}\|\Psi\|^2.
$$
 % ------------- %
Moreover, one can use the well-known recurrent relation for Hermite functions,
 % ------------- %
\begin{equation}
  \sqrt{n+1}\chi_{n+1}(y)-\sqrt{2}y\chi_n(y)+\sqrt{n}\chi_{n-1}(y) = 0.\label{eq:hermite}
\end{equation}
 % ------------- %
Using the Ansatz \eqref{eq:separation} in the definition of $\mathbf{b}_3[\Psi]$ substituting subsequently for $y\chi_n(y)$ from \eqref{eq:hermite}, we get
 % ------------- %
\begin{align*}
  \mathbf{b}_3[\Psi] &= \sum_{n,m\in\mathbb{N}_0} \int_{\mathbb{R}} y\, \mathrm{Re}\big[\gamma\big(\bar \psi_m(0+)+\psi_m(0-)\big)\bar\chi_m(y) \big(\psi_n(0+)-\psi_n(0-)\big)\chi_n(y)\big]\,\mathrm{d}y \\[.15em]
  &= \frac{1}{\sqrt{2}}\sum_{n,m\in\mathbb{N}_0}  \int_{\mathbb{R}} \bar\chi_m(y)\,\mathrm{Re}\big\{\gamma \big(\bar \psi_m(0+)+\psi_m(0-)\big) \\
  &\quad\times \big(\psi_{n}(0+)-\psi_{n}(0-)\big) \big(\sqrt{n+1}\chi_{n+1}(y)+\sqrt{n}\chi_{n-1}(y)\big)\big\}\,\mathrm{d}y.
\end{align*}
 % ------------- %
Using the orthogonality of Hermite functions, $\int_\mathbb{R} \bar\chi_m(y)\chi_n(y)\,\mathrm{d}y = \delta_{mn}$, we further obtain
 % ------------- %
\begin{align*}
  \mathbf{b}_3[\Psi] & =  \frac{1}{\sqrt{2}}\sum_{n\in\mathbb{N}_0} \mathrm{Re}\big\{\gamma \big[\sqrt{n+1}\big(\bar\psi_{n+1}(0+)+\bar\psi_{n+1}(0-)\big)+\sqrt{n} \big(\bar\psi_{n-1}(0+)+\bar\psi_{n-1}(0-)\big)\big]
\\[.15em]
 & \quad\times (\psi_n(0+)-\psi_n(0-)\big)\big\}
\\[.15em]
  & = \frac{1}{\sqrt{2}}\sum_{n\in\mathbb{N}}\mathrm{Re}\big\{\gamma \big[\sqrt{n} \big(\bar\psi_n(0+)+\bar\psi_n(0-)\big) \big(\psi_{n-1}(0+)-\psi_{n-1}(0-)\big)
\\[.15em]
 & \quad +\sqrt{n}\big(\bar\psi_{n-1}(0+)+\bar\psi_{n-1}(0-)\big) \big(\psi_{n}(0+)-\psi_{n}(0-)\big)\big]\big\}
\\[.15em]
  & = \sum_{n\in\mathbb{N}}\sqrt{2n}\,\mathrm{Re}\big\{\mathrm{Re\,}\gamma \big[\bar \psi_n(0+)\psi_{n-1}(0+)-\bar\psi_n(0-)\psi_{n-1}(0-)\big]
\\[.15em]
  &\quad +i\,\mathrm{Im\,}\gamma \big[\bar\psi_{n}(0-)\psi_{n-1}(0+)-\bar\psi_n(0+)\psi_{n-1}(0-)\big]\big\}.
\end{align*}
 % ------------- %
In the first term of the second equality, we have changed the summation index from $n$ to $n-1$; notice that the sum now runs over $\mathbb{N}$, not $\mathbb{N}_0$.
In an analogous way, the other two terms entering the form $\mathbf{a}_{\alpha,\beta,\gamma}[\Psi]$ can be written as
 % ------------- %
\begin{align*}
  \mathbf{b}_1[\Psi] + \mathbf{b}_2[\Psi]  & = \sum_{n\in \mathbb{N}}\sqrt{2n} \,\mathrm{Re\,} \Big[\big(\bar\psi_{n}(0+)-\bar\psi_n(0-)\big) \big(\psi_{n-1}(0+)-\psi_{n-1}(0-)\big)
\\[.15em]
  &\quad + \frac{\alpha\beta+|\gamma|^2}{4} \big(\bar\psi_{n}(0+)+\bar\psi_n(0-)\big) \big(\psi_{n-1}(0+)+\psi_{n-1}(0-)\big)\Big]
\\[.15em]
  & = \sum_{n\in \mathbb{N}}\sqrt{2n}\, \mathrm{Re\,} \Big[\frac{4+\alpha\beta+|\gamma|^2}{4} \big(\bar\psi_n(0+)\psi_{n-1}(0+)+\bar\psi_n(0-)\psi_{n-1}(0-)\big)
\\[.15em]
  &\quad +\frac{\alpha\beta+|\gamma|^2-4}{4} \big(\bar\psi_n(0+)\psi_{n-1}(0-)+\bar\psi_n(0-)\psi_{n-1}(0+)\big)\Big].
\end{align*}
 % ------------- %
The sum of all the three forms can be written elegantly using the Pauli matrices and the identity matrix,
 % ------------- %
\begin{equation}
  \mathbf{b}_1[\Psi] + \mathbf{b}_2[\Psi] +\mathbf{b}_3[\Psi] = \sum_{n\in\mathbb{N}}\frac{\sqrt{n}}{2\sqrt{2}} \,\mathrm{Re\,}\left[\bar{\mathbf{u}}^T  \cdot \Sigma\cdot \mathbf{v} \right],\label{eq:pauli}
\end{equation}
 % ------------- %
where the notation is the same as in Lemma~\ref{lem:tech3}, $\mathbf{u} = \begin{pmatrix}\psi_n(0+)\\ \psi_n(0-)\end{pmatrix}$ and $\mathbf{v} = \begin{pmatrix}\psi_{n-1}(0+)\\ \psi_{n-1}(0-)\end{pmatrix}$, and the numbers $\omega_j$ are given by~\eqref{eq:omega}.
Applying now Lemma~\ref{lem:tech3}, we obtain
 % ------------- %
\begin{align*}
  \mathbf{b}_1[\Psi] + \mathbf{b}_2[\Psi] +\mathbf{b}_3[\Psi] &\leq \sum_{n\in\mathbb{N}}\frac{\sqrt{n}}{4\sqrt{2}}\, \big(|\psi_{n}(0+)|^2+|\psi_{n}(0-)|^2
\\[.15em]
&\quad +|\psi_{n-1}(0+)|^2+|\psi_{n-1}(0-)|^2\big)
\\[.15em]
 & \quad \times \big[|4+\alpha\beta+|\gamma|^2|+\sqrt{(\alpha\beta+|\gamma|^2-4)^2+16|\gamma|^2}\big].
\end{align*}
 % ------------- %
We divide the sum into two parts and in the part containing $\psi_{n-1}$ we raise the summation index by one; in this way we get
 % ------------- %
\begin{align*}
  \mathbf{b}_1[\Psi] + \mathbf{b}_2[\Psi] +\mathbf{b}_3[\Psi] &\leq \sum_{n\in\mathbb{N}_0}\frac{\sqrt{n}+\sqrt{n+1}}{4\sqrt{2}} \big(|\psi_{n}(0+)|^2+|\psi_{n}(0-)|^2\big)
\\[.15em]
 & \quad\times \big[|4+\alpha\beta+|\gamma|^2|+\sqrt{(\alpha\beta+|\gamma|^2-4)^2+16|\gamma|^2}\big].
\end{align*}
 % ------------- %
Furthermore, we employ the inequality $\sqrt{n}+\sqrt{n+1}\leq 2\sqrt{n+\frac{1}{2}}$ which yields
 % ------------- %
\begin{align*}
  \mathbf{b}_1[\Psi] + \mathbf{b}_2[\Psi] +\mathbf{b}_3[\Psi] &\leq \frac{|4+\alpha\beta+|\gamma|^2|+\sqrt{(\alpha\beta+|\gamma|^2-4)^2+16|\gamma|^2}}{2\sqrt{2}}
\\
  & \quad \times \sum_{n\in\mathbb{N}_0} \sqrt{n+\textstyle{\frac{1}{2}}}\, \big(|\psi_{n}(0+)|^2+|\psi_{n}(0-)|^2\big).
\end{align*}
 % ------------- %
Applying now Lemma~\ref{lem:tech2} with $\delta = \sqrt{n+\frac{1}{2}}$ we arrive at
 % ------------- %
\begin{align*}
  \mathbf{b}_1[\Psi] + \mathbf{b}_2[\Psi] +\mathbf{b}_3[\Psi] & \leq \frac{|4+\alpha\beta+|\gamma|^2|+\sqrt{(\alpha\beta+|\gamma|^2-4)^2+16|\gamma|^2}}{2\sqrt{2}}
\\[.15em]
  & \quad \times \sum_{n\in\mathbb{N}_0} \int_{\mathbb{R}}\big[|\psi'(x)|^2+ \big(n+\textstyle{\frac{1}{2}}\big)|\psi(x)|^2\big]\,\mathrm{d}x
\\[.15em]
& = \frac{|4+\alpha\beta+|\gamma|^2|+\sqrt{(\alpha\beta+|\gamma|^2-4)^2+16|\gamma|^2}}{2\sqrt{2}}\, \mathbf{a}_0[\Psi].
\end{align*}
 % ------------- %
The obtained upper bound for $\mathbf{b}_1[\Psi]+\mathbf{b}_2[\Psi]+ \mathbf{b}_3[\Psi]$ leads to the lower bound for $\mathbf{a}_{\alpha,\beta,\gamma}[\Psi]$. The second inequality follows from $\mathbf{a}_0[\Psi]\geq \frac{1}{2}\|\Psi\|^2$ which is a consequence of \eqref{eq:a0}.

The last inequality, the lower bound to $\mathbf{a}_{\alpha,0,\gamma}[\Psi]$ for $\beta=0$, follows from an upper bound to $\mathbf{b}_4[\Psi]$ that can be obtained in a way similar to
\cite[Lemma 2.1]{Sol04b} or \cite[Theorem 4]{EL18} with the use of Lemma~\ref{lem:tech1}.
\end{proof}
 % ------------- %

The form $\mathbf{a}_{\alpha,\beta,\gamma}[\Psi]$ is closed on $D$. First, the form $\mathbf{a}_{0}[\Psi]$ is closed, as proven in \cite{BE17b}. The closedness of the full form can be obtained by the same reasoning as in \cite[Proposition 2.2]{BE17b} using the bounds in Theorem~\ref{thm:bounds}; as in the said paper, the argument applies only if the form is semibounded from below.

%%%%%%%%%%%%%%%%%%%%%%%%%%%%%
\section{The Jacobi operator}\label{sec:jacobi}
%%%%%%%%%%%%%%%%%%%%%%%%%%%%%

The second main step of our argument is the construction of the Jacobi operator associated with $\mathbf{H}_{\alpha,\beta,\gamma}$. Let us begin with the case $\beta \ne 0$. First of all, we rewrite the matching conditions \eqref{eq:coupl1} and~\eqref{eq:coupl2} in an equivalent form \cite{EG99},
 % ------------- %
\begin{align*}
 \frac{\partial \Psi}{\partial x}(0+,y) &= \frac{y}{4\beta} \big[(\alpha \beta+|\gamma|^2+4+4\,\mathrm{Re\,}\gamma)\Psi(0+,y)
\\[.15em]
  & \quad +(\alpha \beta+ |\gamma|^2 -4-4i\,\mathrm{Im\,}\gamma)\Psi(0-,y)\big],
\\[.15em]
 \frac{\partial \Psi}{\partial x}(0-,y) &= \frac{y}{4\beta} \big[(-\alpha \beta-|\gamma|^2+4 - 4 i\, \mathrm{Im\,}\gamma)\Psi(0+,y)
\\[.15em]
 & \quad +(-\alpha \beta-|\gamma|^2-4+4\,\mathrm{Re\,}\gamma)\Psi(0-,y)]\big.
\end{align*}
 % ------------- %
Mimicking the construction in \cite{NS06,EL18} we use the Ansatz \eqref{eq:separation} and the recurrent relation \eqref{eq:hermite}, multiply both equations from the left by $\bar\chi_m(y)$ and integrate over $\mathbb{R}$ obtaining thus conditions for the coefficient functions,
 % ------------- %
\begin{subequations}
\label{eq:coupl5&6}
\begin{align}
  \frac{\partial \psi_m}{\partial x}(0+) &=  \frac{1}{4\sqrt{2}\beta} \big[(\alpha \beta+|\gamma|^2+4+4\mathrm{Re\,}\gamma)(\psi_{m-1}(0+)\sqrt{m}+\psi_{m+1}(0+)\sqrt{m+1})\nonumber
\\[.15em]
& +(\alpha \beta+|\gamma|^2-4-4i\mathrm{Im\,}\gamma)(\psi_{m-1}(0-)\sqrt{m}+\psi_{m+1}(0-)\sqrt{m+1})\big],\label{eq:coupl5}
  \\[.15em]
  \frac{\partial \psi_m}{\partial x}(0-) &=  \frac{1}{4\sqrt{2}\beta} \big[(-\alpha \beta-|\gamma|^2+4-4i\mathrm{Im\,}\gamma)(\psi_{m-1}(0+)\sqrt{m}+\psi_{m+1}(0+)\sqrt{m+1})\nonumber
\\[.15em]
& +(-\alpha \beta-|\gamma|^2-4+4\mathrm{Re\,}\gamma)(\psi_{m-1}(0-)\sqrt{m}+\psi_{m+1}(0-)\sqrt{m+1})\big].\label{eq:coupl6}
\end{align}
\end{subequations}
 % ------------- %
We seek solutions of the form
 % ------------- %
\begin{subequations}
\label{eq:phi1&2}
\begin{align}
  \phi_n(x,\Lambda) &= k_1(\Lambda,n)\, \mathrm{e}^{-\zeta_n(\Lambda)x},\quad x>0,\label{eq:phi1}\\[.15em]
 \phi_n(x,\Lambda) &= k_2(\Lambda,n)\, \mathrm{e}^{\zeta_n(\Lambda)x},\quad  \;\;\, x<0,\label{eq:phi2}
\end{align}
\end{subequations}
 % ------------- %
where $\zeta_n(\Lambda) := \sqrt{n+\frac12-\Lambda}$, to the equation
 % ------------- %
\begin{equation}
  - \phi_n''(x)+\left(n+\frac12-\Lambda\right) \phi_n(x) = 0\,,\quad n\in \mathbb{N}_0, \label{eq-un}
\end{equation}
 % ------------- %
with $\phi_n \in H^2(-\infty,0)\oplus H^2(0,\infty)$ satisfying the conditions \eqref{eq:coupl5&6}. In the definition of $\zeta_n(\Lambda)$ we take the branch of the square root which is analytic in $\mathbb{C}\backslash [n+\frac12,\infty)$ and for a number~$\Lambda$ from this set it holds
 % ------------- %
$$
  \mathrm{Re\,}\zeta_n(\Lambda) >0\,,\quad \mathrm{Im\,}\zeta_n(\Lambda)\cdot \mathrm{Im\,}\Lambda<0.
$$
 % ------------- %
Substituting the Ans\"atze \eqref{eq:phi1&2} into \eqref{eq:coupl5&6} we get
\begin{align*}
  -\zeta_n(\Lambda) k_1(\Lambda,n) &= \frac{1}{4\sqrt{2}\beta} \big[(\alpha \beta+|\gamma|^2+4+4\mathrm{Re\,}\gamma)(k_1(\Lambda, n-1)\sqrt{n}\nonumber
\\[.15em]
& \quad +k_1(\Lambda, n+1)\sqrt{n+1}) +(\alpha \beta+|\gamma|^2-4-4i\mathrm{Im\,}\gamma) %\label{eq:coupl7}
  \\[.15em]
& \quad \times (k_2(\Lambda, n-1)\sqrt{n}+k_2(\Lambda, n+1)\sqrt{n+1})\big],
\\[.15em]
  \zeta_n(\Lambda) k_2(\Lambda,n) &= \frac{1}{4\sqrt{2}\beta} \big[(-\alpha \beta-|\gamma|^2+4-4i\mathrm{Im\,}\gamma)(k_1(\Lambda, n-1)\sqrt{n}\sqrt{n+1})\nonumber
\\[.15em]
& \quad +k_1(\Lambda, n+1) +(-\alpha \beta-|\gamma|^2-4+4\mathrm{Re\,}\gamma)(k_2(\Lambda, n-1)\sqrt{n}
\\[.15em]
& \quad +k_2(\Lambda, n+1)\sqrt{n+1})\big].
%\label{eq:coupl8}
\end{align*}
 % ------------- %
Adding and subtracting the previous two equations we get
 % ------------- %
\begin{align*}
  \zeta_n(\Lambda) \big[k_2(\Lambda,n)-k_1(\Lambda,n)] &= \frac{2}{\sqrt{2}\beta}[(k_1(\Lambda,n-1)-k_2(\Lambda,n-1))\sqrt{n} \nonumber
\\[.15em]
& +(k_1(\Lambda,n+1)-k_2(\Lambda,n+1))\sqrt{n+1}\big]\nonumber
\\[.15em]
& +\frac{\bar\gamma}{\sqrt{2}\beta} \big[(k_1(\Lambda,n-1)+k_2(\Lambda,n-1))\sqrt{n}\nonumber
\\[.15em]
& +(k_1(\Lambda,n+1)+k_2(\Lambda,n+1))\sqrt{n+1}\big], %\label{eq:coupl11}
\\[.15em]
  -\zeta_n(\Lambda) \big[k_1(\Lambda,n)+k_2(\Lambda,n)] &= \frac{\gamma}{\sqrt{2}\beta}[(k_1(\Lambda,n-1)-k_2(\Lambda,n-1))\sqrt{n}\nonumber
\\[.15em]
& +(k_1(\Lambda,n+1)-k_2(\Lambda,n+1))\sqrt{n+1}\big]\nonumber
\\[.15em]
& + \frac{\alpha\beta+|\gamma|^2}{2\sqrt{2}\beta} \big[(k_1(\Lambda,n-1)+k_2(\Lambda,n-1))\sqrt{n}\nonumber
\\[.15em]
& +(k_1(\Lambda,n+1)+k_2(\Lambda,n+1))\sqrt{n+1}\big]. %\label{eq:coupl12}
\end{align*}
 % ------------- %
To simplify these relations, we define
 % ------------- %
$$
C_\pm(n) := k_1(\Lambda,n)\pm k_2(\Lambda,n),
$$
 % ------------- %
then the previous two equations can be rewritten as
 % ------------- %
$$
  -\zeta_n(\Lambda)\begin{pmatrix}C_-(n)\\C_+(n)\end{pmatrix} = \begin{pmatrix}\frac{2}{\sqrt{2}\beta} & \frac{\bar \gamma}{\sqrt{2}\beta}\\ \frac{\gamma}{\sqrt{2}\beta}& \frac{\alpha\beta+|\gamma|^2}{2\sqrt{2}\beta}\end{pmatrix}\begin{pmatrix}C_-(n-1)\sqrt{n}+C_-(n+1)\sqrt{n+1}\\C_+(n-1)\sqrt{n}+C_+(n+1)\sqrt{n+1}\end{pmatrix}
$$
 % ------------- %
Computing the eigenvalues and eigenspaces of the matrix involved in the previous equation and putting
 % ------------- %
\begin{subequations}
\label{eq:jacobic+-1&2}
\begin{align}
\hspace{-1.25em} C_{n,1} &:= \left(n+\textstyle{\frac{1}{2}}\right)^{-1/4} \big[-4\gamma C_-(n)+ \big(4-\alpha \beta-|\gamma|^2
+\sqrt{(\alpha\beta+|\gamma|^2-4)^2+16|\gamma|^2}\big)C_+(n)\big],\label{eq:jacobic+-1}
\\[.15em]
\hspace{-1.25em} C_{n,2} &:= \left(n+\textstyle{\frac{1}{2}}\right)^{-1/4} \big[4\gamma C_-(n)+ \big(-4+\alpha \beta+|\gamma|^2
+\sqrt{(\alpha\beta+|\gamma|^2-4)^2+16|\gamma|^2}\big)C_+(n)\big], \label{eq:jacobic+-2}
\end{align}
\end{subequations}
 % ------------- %
we find that $C_{n,1}$ satisfies the same equation as, for instance, in \cite{EL18}, namely
 % ------------- %
\begin{equation}
  (n+1)^{1/2}\left(n+\textstyle{\frac{3}{2}}\right)^{1/4}C_{n+1,1}+2\mu \left(n+\textstyle{\frac{1}{2}}\right)^{1/4}\zeta_n(\Lambda)C_{n,1}
  +n^{1/2}\left(n-\textstyle{\frac{1}{2}}\right)^{1/4}C_{n-1,1} = 0\label{eq:jacobi}
\end{equation}
 % ------------- %
with $\mu=\mu_1$ given by
 % ------------- %
\begin{subequations}
\label{eq:mu_1&2}
\begin{equation}
  \mu_1 := \frac{2\sqrt{2}\beta}{4+\alpha\beta+|\gamma|^2-\sqrt{(\alpha\beta+|\gamma|^2-4)^2+16|\gamma|^2}}\,. \label{eq:mu_1}
\end{equation}
 % ------------- %
Similarly, $C_{n,2}$ satisfies the same equation, the only difference is that with $\mu$ equals now to
 % ------------- %
\begin{equation}
  \mu_2 := \frac{2\sqrt{2}\beta}{4+\alpha\beta+|\gamma|^2+\sqrt{(\alpha\beta+|\gamma|^2-4)^2+16|\gamma|^2}}\,. \label{eq:mu_2}
\end{equation}
\end{subequations}
 % ------------- %

With the above mentioned symmetry $(\alpha,\beta)\to (-\alpha,-\beta)$ in mind, we can without the loss of generality take $\beta>0$. Then we easily find that for $\alpha>0$ it holds $\mu_1\geq \mu_2>0$, for $\alpha<0$ it holds $\mu_1<0<\mu_2$ and for $\alpha =0$ the parameter $\mu_1$ diverges.

The general solution of the equation~\eqref{eq-un} can be written in the form of linear combination of functions corresponding to the above-mentioned eigenspaces $\phi_n = C_{n,1}\eta_n^{(1)}+C_{n,2}\eta_n^{(2)}$, where
 % ------------- %
\begin{eqnarray*}
  \eta_n^{(1)} = \left\{\begin{matrix} (n+1/2)^{1/4}(\omega_3-\sqrt{\omega_1^2+\omega_2^2+\omega_3^2})\, \mathrm{e}^{-\zeta_n(\Lambda) x} & x>0\\
                                       (n+1/2)^{1/4}(\omega_1+i \omega_2)\, \mathrm{e}^{\zeta_n(\Lambda) x} & x<0\end{matrix}\right.\\
  \eta_n^{(2)} = \left\{\begin{matrix} (n+1/2)^{1/4}(\omega_3+\sqrt{\omega_1^2+\omega_2^2+\omega_3^2})\, \mathrm{e}^{-\zeta_n(\Lambda) x} & x>0\\
                                       (n+1/2)^{1/4}(\omega_1+i \omega_2)\, \mathrm{e}^{\zeta_n(\Lambda) x} & x<0\end{matrix}\right.
\end{eqnarray*}
 % ------------- %
with $\omega_j$ as in~\eqref{eq:omega}. There exist constants $C_1$, $C_2$ not dependent on $n$ so that
 % ------------- %
\begin{equation}
 C_1(\lambda) \leq \|\eta_n^{(1)}\|^2+ \|\eta_n^{(2)}\|^2 \leq C_2(\Lambda) \quad \mathrm{for\ all\ }n\in \mathbb{N}_0\label{eq:etabound}
\end{equation}
 % ------------- %
provided $\gamma\ne 0$ or $\alpha\beta\ne 4$. In the degenerate case $\gamma= 0$ and $\alpha\beta= 4$ it holds $\mu_1 = \mu_2$ and we have, for instance,
  % ------------- %
\begin{eqnarray*}
  \eta_n^{(1)} = \left\{\begin{matrix} (n+1/2)^{1/4}\, \mathrm{e}^{-\zeta_n(\Lambda) x} & \quad x>0\\
                                       0 & \quad x<0\end{matrix}\right.\,,\\
  \eta_n^{(2)} = \left\{\begin{matrix} 0 & \quad x>0\\
                                       (n+1/2)^{1/4}\, \mathrm{e}^{\zeta_n(\Lambda) x} & \quad x<0\end{matrix}\right.\,.
\end{eqnarray*}
 % ------------- %

The procedure for $\beta = 0$ is similar. We rewrite the coupling conditions as
 % ------------- %
 \begin{subequations}
\label{eq:beta=0coup1&2}
\begin{align}
  \Big(1+\frac{\bar \gamma}{2}\Big)\Psi(0+, y)-\Big(1-\frac{\bar \gamma}{2}\Big)\Psi(0-,y) &= 0,\label{eq:beta=0coup1}\\
  \Big(1-\frac{\gamma}{2}\Big)\frac{\partial \Psi}{\partial x}(0+,y)-\Big(1+\frac{\gamma}{2}\Big)\frac{\partial \Psi}{\partial x}(0-,y)& = \frac{\alpha y}{2} \big[\Psi(0+,y)+\Psi(0-,y)\big].\label{eq:beta=0coup2}
\end{align}
\end{subequations}
 % ------------- %
From \eqref{eq:beta=0coup1} we infer that there exist a sequence of coefficients $C_n$ such that
 % ------------- %
\begin{align*}
  k_1(\Lambda,n) &= \Big(n+\frac{1}{2}\Big)^{1/4}\Big(1-\frac{\bar\gamma}{2}\Big)C_n, \\[.15em]
  k_2(\Lambda,n) &= \Big(n+\frac{1}{2}\Big)^{1/4}\Big(1+\frac{\bar\gamma}{2}\Big)C_n,
\end{align*}
 % ------------- %
where $k_j$, $j=1,2$, are the coefficients in the solutions $\phi_j$, $j=1,2$, of~\eqref{eq:phi1} and~\eqref{eq:phi2}. After multiplying \eqref{eq:beta=0coup2} by $\bar\chi_m(y)$, integrating it using the orthogonality of Hermite functions, and substituting for $k_j$, $j=1,2$, we obtain again the equation~\eqref{eq:jacobi}, now with the parameter in the middle term taking the value $\mu := \frac{4+|\gamma|^2}{2\sqrt{2}\alpha}$. This is the departing point for the spectral analysis.

%%%%%%%%%%%%%%%%%%%%%%%%%%%%%%%%%%%%%%%%
\section{Self-adjointness of the Hamiltonian}\label{sec:sa}
%%%%%%%%%%%%%%%%%%%%%%%%%%%%%%%%%%%%%%%%
In this section, we prove that the Hamiltonian $\mathbf{H}_{\alpha,\beta,\gamma}$ is self-adjoint. The strategy of our proof is inspired by the proof of \cite[Thm. 4.1]{NS06}.

The vector $\Psi$ can be using~\eqref{eq:separation} identified with the sequence $\{\psi_n\}_{n\in\mathbb{N}_0}$ (further denoted only by $\{\psi_n\}$). The Hamiltonian $\mathbf{H}_{\alpha,\beta,\gamma}$ defined in~\eqref{model} can be alternatively defined via its action
$$
  \mathbf{H}_{\alpha,\beta,\gamma} \{\psi_n\} = \{-\psi_n''+(n+1/2)\psi_n\}\,.
$$
and the domain $D_{\alpha,\beta,\gamma}$. We say that the vector $\{\psi_n\}\in D_{\alpha,\beta,\gamma}$ if and only if $\psi_n\in \{u,u\upharpoonright_{\mathbb{R}_{\pm}}\in H^2(\mathbb{R}_{\pm})\} = \mathcal{W}$,  conditions~\eqref{eq:coupl5&6} are satisfied and  $\sum_{n\in \mathbb{N}_0}\| $$ -\psi_n''(x)+(n+1/2)\psi_n(x)\|^2 <\infty$. This construction identifies the $L^2(\mathbb{R}^2)$ space of $\Psi$ with the space $\mathfrak{H} = \ell^2\otimes L^2(\mathbb{R})$. Moreover, we define $\mathbf{H}^0_{\alpha,\beta,\gamma}:=\mathbf{H}_{\alpha,\beta,\gamma} \upharpoonright D^0_{\alpha,\beta,\gamma}$, where the domain $D^0_{\alpha,\beta,\gamma}\subset D_{\alpha,\beta,\gamma}$ consists of all ${\psi_n} \in D_{\alpha,\beta,\gamma}$ with finite number of non-zero elements.

\begin{theorem}
The Hamiltonian $\mathbf{H}_{\alpha,\beta,\gamma}$ is self-adjoint and coincides with the closure of $\mathbf{H}^0_{\alpha,\beta,\gamma}$.
\end{theorem}
\begin{proof}
We give the proof for the general case. We proceed similarly to \cite[Thm. 4.1]{NS06}. First we prove that
 % ------------- %
\begin{equation}\label{eq:h=h0*}
\mathbf{H}_{\alpha,\beta,\gamma} = {\mathbf{H}^0_{\alpha,\beta,\gamma}}^*\,.
\end{equation}
 % ------------- %
The inclusion $\subset$ is trivial. We will prove the opposite inclusion. Suppose $\{v_n\}\subset \mathrm{Dom\,}({\mathbf{H}^0_{\alpha,\beta,\gamma}}^*)$ and $ {\mathbf{H}^0_{\alpha,\beta,\gamma}}^* \{v_n\} = \{w_n\}$. The definition of the adjoint operator yields
\begin{equation}
  \sum_{n\in\mathbb{N}_0}\int_{\mathbb{R}}(-u_n''+(n+1/2)u_n)\overline{v_n}\,\mathrm{d}x = \sum_{n\in\mathbb{N}_0}\int_{\mathbb{R}}u_n\overline{w_n}\,\mathrm{d}x\,.\label{eq:sa:adjoint}
\end{equation}
Let us fix $n_0\in \mathbb{N}$ and take the element $\{u_n\}\in D^0_{\alpha,\beta,\gamma}$ such that for a fixed $n_0\in \mathbb{N}_0$ it holds $u_n \equiv 0$ for $n\ne n_0$ and $u_{n_0}\in \mathcal{W}$, $u_{n_0}=0$ in the vicinity of $x=0$. Then equation~\eqref{eq:sa:adjoint} applied to all such $\{u_n\}$ implies that if $v\in \mathrm{Dom\,}({\mathbf{H}^0_{\alpha,\beta,\gamma}}^*)$, $w_n\in \mathcal{W}$ and $w_n = -v_n''+(n+1/2)v_n$ for all $n\in \mathbb{N}_0$. Moreover, $w_n\in l^2(\mathbb{N}_0)$ implies that $\sum_{n\in \mathbb{N}_0}\|-\psi_n''(x)+(n+1/2)\psi_n(x)\|^2 <\infty$.

As the final part of the proof of the equality~\eqref{eq:h=h0*}, we prove that the coupling conditions~\eqref{eq:coupl5&6} are satisfied for $\{v_n\}$. We choose $\{u_n\}$ such that for a fixed $n_0$
 % ------------- %
 \begin{subequations} \label{eq:ccu}
 \begin{align}
  u_{n_0\pm 1} (0\pm) = 0\,,\quad u'_{n_0}(0\pm)=0\,,\quad u_{n_0}(0\pm) = h_{\pm}\,,
\\
	u_n(x)\equiv 0\quad \mathrm{for\ all\ }n\not\in\{n_0-1,n_0, n_0+1\}\,,
 \\
  \ \begin{pmatrix}u'_{n_0+ 1}(0+)\\-u'_{n_0+ 1}(0-)\end{pmatrix} = \frac{\sqrt{n_0+1}}{4\sqrt{2}\beta}\Sigma\begin{pmatrix}h_+\\h_-\end{pmatrix}\,,\quad
  \begin{pmatrix}u'_{n_0- 1}(0+)\\-u'_{n_0- 1}(0-)\end{pmatrix} = \frac{\sqrt{n_0}}{4\sqrt{2}\beta}\Sigma\begin{pmatrix}h_+\\h_-\end{pmatrix}\,.
  \end{align}
\end{subequations}
 % ------------- %
Clearly, such $\{u_n\}$ satisfies~\eqref{eq:coupl5&6} and belongs to $D^0_{\alpha,\beta,\gamma}$. The equation~\eqref{eq:sa:adjoint} implies by Green's identity that
 % ------------- %
$$
   \sum_{n\in \mathbb{N}_0}[u_n'(x)\overline{v_n}(x)-u_n(x)\overline{v'_n}(x)]_{0-}^{0+} = 0\,.
$$
 % ------------- %
From this, \eqref{eq:ccu} and the hermiticity of the matrix $\Sigma$ it follows that $\{v_n\}$ satisfies~\eqref{eq:coupl5&6}. This proves the equation~\eqref{eq:h=h0*}.

In the second part of the proof, we find that the deficiency indices of the operator $\mathbf{H}_{\alpha,\beta,\gamma}$ are zero. We will prove that the only solution to the equation
 % ------------- %
\begin{equation}
  \mathbf{H}_{\alpha,\beta,\gamma} V - i V = 0 \label{eq:defind}
\end{equation}
 % ------------- %
is trivial. The equation~\eqref{eq:defind} is the equation~\eqref{eq-un} for $\Lambda=i$. For its general solution $V = \{\phi_n\} = \{C_{n,1}\eta_n^{(1)}+C_{n,2}\eta_n^{(2)}\}$ we have due to $\int_{\mathbb{R}}\eta_n^{(1)}\overline{\eta_n^{(2)}}\,\mathrm{d}x = 0$ that
$\|\phi_n\|^2 = |C_{n,1}|^2\|\eta_n^{(1)}\|^2+|C_{n,2}|^2\|\eta_n^{(2)}\|^2$. Hence one can deal with the two subspaces separately. Using this fact and~\eqref{eq:etabound}, the construction of large $n$ asymptotics of the coefficients $C_{n,1}$ and $C_{n,2}$ follows similarly to \cite{NS06}. In the rest of the proof, we will use the abbreviation $C_n$ for both $C_{n,1}$ and $C_{n,2}$. From Lemma~\ref{lem:C:asymptotics} we obtain for the particular cases the following asymptotics.
\begin{enumerate}
\item For $\mu^2<1$ it holds $|C_n^{\pm}|^2 \sim n^{-1\pm\frac{\mu}{\sqrt{1-\mu^2}}}$. The sequence $C_n^{+}$ is not in $\ell^2$, hence the corresponding solution does not belong to $\mathfrak{H}$. Furthermore, we use the identity~\eqref{eq:Cnidentity2} for $C_n^{-}$. Its \emph{rhs} vanishes for $N\to \infty$. All the terms on the \emph{lhs} have the same sign (note that $\mathrm{Im\,}\zeta_n(i)<0$), hence all  $C_n^{-}=0$.
\item For $\mu^2>1$ it holds $|C_n^{\pm}|^2 \sim (\sqrt{\mu^2-1}-\mu)^{\pm 2n} n^{-1} =  (\sqrt{\mu^2-1}+\mu)^{\mp 2n} n^{-1}$. Hence one solution exponentially grows and the other exponentially decays. The growing solution is not in $\ell^2$, the decaying one is zero due to the identity~\eqref{eq:Cnidentity2}.
\item For $\mu=\pm 1$ it holds $|C_n^{\pm}|^2 \sim n^{-1/2} \mathrm{e}^{\pm 2 \sqrt{2n}}$. The growing sequence is not in $\ell^2$, the decaying sequence vanishes due to~\eqref{eq:Cnidentity2}.
\end{enumerate}
Hence $V\equiv 0$ and the operator is self-adjoint.
\end{proof}

%%%%%%%%%%%%%%%%%%%%%%%%%%%%%%%%%%%%%%%%
\section{Absolutely continuous spectrum}\label{sec:ac}
%%%%%%%%%%%%%%%%%%%%%%%%%%%%%%%%%%%%%%%%

Having obtained the equation \eqref{eq:jacobi} having the Jacobi operator structure, we can proceed in a way similar to that of \cite{NS06, EL18} to obtain the absolutely continuous spectrum of our model Hamiltonian. The proofs of the claims made below are non-trivial, being based on the existence/nonexistence of subordinate solutions, but we are not going to present them as the argument is a straightforward adaptation of the particular cases treated in the
mentioned publications, cf. especially Theorem~3.1 and Appendix~A in \cite{NS06}.

To state the result, let us first introduce the needed operators:
 % ------------- %
\begin{definition}
Given a sequence $\{\omega_n\}_{n\in\mathbb{N}_0}$ the operators $\mathcal{D}\equiv \mathcal{D} \{\omega_n\}$ of multiplication by a this sequence and of right shift $\mathcal{S}$ mapping $\ell^2(\mathbb{N}_0) \to \ell^2(\mathbb{N}_0)$ act as
 % ------------- %
\begin{align*}
\mathcal{D} \{\omega_n\} &: \{r_0, r_1,\dots\} \mapsto \{\omega_0 r_0, \omega_1 r_1,\dots\}, \\[.15em]
\mathcal{S} &: \{r_0,r_1, \dots\} \mapsto \{0,r_0,r_1,\dots\}.
\end{align*}
 % ------------- %
Using them we define
 % ------------- %
\begin{align*}
P &:= \mathcal{D}(d_n), \quad d_n := n^{1/2}(n+\textstyle{\frac12})^{1/4}(n-\textstyle{\frac12})^{1/4}, \\[.15em]
Y &:= \mathcal{D}\{(n+\textstyle{\frac12})^{1/2}\zeta_n(\Lambda)\}, \\[.15em]
Y_0 &:= \mathcal{D}\{n+\textstyle{\frac12}\}, \\[.15em]
\mathcal{J} (\Lambda,\mu) &:= P \mathcal{S}+\mathcal{S}^* P + 2\mu Y\,. \\[.15em]
\mathcal{J}_0 (\mu) &:= P \mathcal{S}+\mathcal{S}^* P + 2\mu Y_0\,.
\end{align*}
 % ------------- %
\end{definition}
 % ------------- %

\smallskip

\noindent The Jacobi matrix corresponding to the equation~\eqref{eq:jacobi} multiplied by $(n+\textstyle{\frac12})^{1/4}$ is $\mathcal{J}(\Lambda, \mu)$. The difference $\mathcal{J}(\Lambda, \mu)-(\mathcal{J}_0(\mu)-\mu\Lambda)$ is a compact operator (see, e.g., \cite{NS06}), hence we can study the absolutely continuous spectrum of the operator $\mathcal{J}_0(\mu)$. Its spectrum is determined by the value of the parameter $\mu$ and its character changes abruptly at $\mu=1$.
 % ------------- %
\begin{theorem}\label{thm:acj0}
We have
 % ------------- %
\begin{align*}
	\sigma(\mathcal{J}_0(\mu)) = \sigma_\mathrm{ac}(\mathcal{J}_0(\mu)) &= \mathbb{R} \quad \mathrm{for}\;\; 0<\mu<1, \\[.15em]
	\sigma(\mathcal{J}_0(1)) = \sigma_\mathrm{ac}(\mathcal{J}_0(1)) &= [0,\infty), \\[.15em]
	\sigma(\mathcal{J}_0(\mu)) = \sigma_\mathrm{disc}(\mathcal{J}_0(\mu)) & \subset(0,\infty) \quad \mathrm{for}\;\; \mu>1, \\[.15em]
	\mathfrak{m}_\mathrm{ac}(E,\mathcal{J}_0(\mu)) &= 1\quad\mathrm{a.e.\ on }\;\; \sigma(\mathcal{J}_0(\mu))
\end{align*}
 % ------------- %
where $\mathfrak{m}_\mathrm{ac}$ denotes the multiplicity function of the absolutely continuous spectrum.
\end{theorem}
 % ------------- %

\smallskip

\noindent Repeating the reasoning of \cite[Sec.~6--8]{NS06} one can check that the spectrum of our Hamiltonian is the union of spectra of the `free' operator and that of the above Jacobi operator:
 % ------------- %
\begin{theorem}\label{thm:ach}
The spectrum of $\mathbf{H}_{0,0,0}$ is purely absolutely continuous, $\sigma(\mathbf{H}_{0,0,0}) = \left[\frac{1}{2},\infty\right)$. For the `full' operator we have
 % ------------- %
\begin{align*}
	\sigma_{\mathrm{ac}}(\mathbf{H}_{\alpha,\beta,\gamma} ) &= \sigma_\mathrm{ac}(\mathbf{H}_{0,0,0})
\\
&\cup \sigma_\mathrm{ac}\Big(\mathcal{J}_0\Big(\frac{2\sqrt{2}\beta}{4+\alpha\beta+|\gamma|^2-\sqrt{(\alpha\beta+|\gamma|^2-4)^2+16|\gamma|^2}}\Big)\Big)
\\
&\cup \,\sigma_\mathrm{ac}\Big(\mathcal{J}_0\Big(\frac{2\sqrt{2}\beta}{4+\alpha\beta+|\gamma|^2+\sqrt{(\alpha\beta+|\gamma|^2-4)^2+16|\gamma|^2}}\Big)\Big)\,,\quad \beta\ne 0, \\[.15em]
	\mathfrak{m}_\mathrm{ac}(E,\mathbf{H}_{\alpha,\beta,\gamma}) &= \mathfrak{m}_\mathrm{ac}(E,\mathbf{H}_{0,0,0})
\\
&+ \mathfrak{m}_\mathrm{ac}\Big(E,\mathcal{J}_0\Big(\frac{2\sqrt{2}\beta}{4+\alpha\beta+|\gamma|^2-\sqrt{(\alpha\beta+|\gamma|^2-4)^2+16|\gamma|^2}}\Big)\Big)
\\
&+ \mathfrak{m}_\mathrm{ac}\Big(E,\mathcal{J}_0\Big(\frac{2\sqrt{2}\beta}{4+\alpha\beta+|\gamma|^2+\sqrt{(\alpha\beta+|\gamma|^2-4)^2+16|\gamma|^2}}\Big)\Big) \,,\quad \beta\ne 0, \\[.15em]
	\sigma_{\mathrm{ac}}(\mathbf{H}_{\alpha,0,\gamma} ) &= \sigma_\mathrm{ac}(\mathbf{H}_{0,0,0})\cup \sigma_\mathrm{ac}\Big(\mathcal{J}_0\Big(\frac{4+|\gamma|^2}{2\sqrt{2}\alpha}\Big)\Big), \\[.15em]
	\mathfrak{m}_\mathrm{ac}(E,\mathbf{H}_{\alpha,0,\gamma}) &= \mathfrak{m}_\mathrm{ac}(E,\mathbf{H}_{0,0,0}) + \mathfrak{m}_\mathrm{ac}\Big(E,\mathcal{J}_0\Big(\frac{4+|\gamma|^2}{2\sqrt{2}\alpha}\Big)\Big) \,.
\end{align*}
 % ------------- %
\end{theorem}
 % ------------- %

\noindent The point where the expressions \eqref{eq:mu_1&2} equal to one, or where $\frac{4+|\gamma|^2}{2\sqrt{2}\alpha}=1$ if $\beta=0$, characterize the manifold in the parameter space at which the spectral transition occurs.

%%%%%%%%%%%%%%%%%%%%%%%%%%%
\section{Discrete spectrum}\label{sec:dicrete}
%%%%%%%%%%%%%%%%%%%%%%%%%%%

We know from the particular cases of the $\delta$ and $\delta'$ coupling that in the situation naturally called \emph{subcritical}, corresponding here to $\mu>1$, the discrete spectrum is nonempty whenever the interaction is `switched on'. Let us look at it now in more detail. We denote by $N_+(\lambda, \mathbf{A}):=\mathrm{dim\,}E^{\mathbf{A}}(\lambda,\infty)$ and  $N_-(\lambda, \mathbf{A}):=\mathrm{dim\,}E^{\mathbf{A}}(-\infty,\lambda)$ the dimension of the spectral projection on the intervals $(\lambda,\infty)$ and $(-\infty,\lambda)$, respectively; the symbol $E^{\mathbf{A}}(\cdot)$ denotes here the spectral measure of the operator $\mathbf{A}$. These quantities are introduced in the subcritical situation ($\mu>1$) only.

To state the results we again need an auxiliary notion. By $\mathbf{J}(\varepsilon)$ we denote the Jacobi operator in $\ell^2(\mathbb{N}_0)$ such that the only nonzero entries in its matrix representation are
 % ------------- %
$$
  j_{n,n-1}(\varepsilon) = j_{n-1,n}(\varepsilon) = \frac{n^{1/2}}{2(n+\varepsilon)^{1/4}(n-1+\varepsilon)^{1/4}}\,,\quad n\in\mathbb{N}.
$$
 % ------------- %
We will also employ another Jacobi operator in $\ell^2(\mathbb{N}_0)$, denoted as $\mathbf{J}_0$, not to be mixed up with $\mathcal{J}_0(\mu)$ used above, with the non-zero entries
 % ------------- %
$$
  j_{n,n-1} = j_{n-1,n} = \frac{1}{2(1-n^{-1})^{1/4}}\,,\quad n\in\mathbb{N}\backslash \{1\}\,.
$$
 % ------------- %
Then we have the following result:
 % ------------- %
\begin{theorem}\label{thm:discrete1}
Let the numbers \eqref{eq:mu_1&2} satisfy $\mu_1>1$ or $\mu_2>1$, then
 % ------------- %
$$
  N_-\Big(\frac{1}{2}-\varepsilon,\mathbf{H}_{\alpha,\beta,\gamma}\Big) = N_1(\varepsilon)+N_2(\varepsilon)
$$
 % ------------- %
holds for $\beta\ne 0$, where
 % ------------- %
$$ %\begin{eqnarray*}
  N_j(\varepsilon) := N_+(\mu_j,\mathbf{J}(\varepsilon)) = N_-(-\mu_j,\mathbf{J}(\varepsilon))\,, \quad j=1,2\,.
$$ %\end{eqnarray*}
 % ------------- %
Furthermore, if $\beta=0$ and $\frac{4+|\gamma|^2}{2\sqrt{2}\alpha}>1$, we have
 % ------------- %
$$
  N_-\Big(\frac{1}{2}-\varepsilon,\mathbf{H}_{\alpha,0,\gamma}\Big) =  N_+\Big(\frac{4+|\gamma|^2}{2\sqrt{2}\alpha},\mathbf{J}(\varepsilon)\Big) = N_-\Big(-\frac{4+|\gamma|^2}{2\sqrt{2}\alpha},\mathbf{J}(\varepsilon)\Big)\,.
$$
 % ------------- %
\end{theorem}
 % ------------- %
\begin{proof}
The argument follows closely \cite[Thm. 3.1]{Sol04b} and \cite[Thm. 10]{EL18}. We will give the proof for $\beta \ne 0$, the procedure for $\beta = 0$ is similar. We remind that the subspaces $D_j\subset D$, $j=1,2$, defined in Section~\ref{sec:quadratic}, are the subspaces of functions for which the Jacobi equation holds with $\mu_j,\:j=1,2$, cf. the formul\ae\ \eqref{eq:jacobic+-1}--\eqref{eq:jacobic+-2}. By the variational principle,
 % ------------- %
$$
   N_-\big(\textstyle{\frac{1}{2}}-\varepsilon,\left.\mathbf{H}_{\alpha,\beta,\gamma}\right|_{D_j}\big) = \mathop{\mathrm{max}}_{\mathcal{F}\in\mathfrak{F}_j(\varepsilon)}\,\mathrm{dim\,}\mathcal{F},
$$
 % ------------- %
where $\mathfrak{F}_j(\varepsilon)$ is the family of all subspaces $\mathcal{F}\subset D_j$ such that
 % ------------- %
\begin{equation}
  \mathbf{a}_{\alpha,\beta,\gamma}[\Psi] - \big(\textstyle{\frac{1}{2}}-\varepsilon\big)\|\Psi\|^2_{L^2(\mathbb{R}^2)}<0 \label{eq:condspec}
\end{equation}
 % ------------- %
We define
 % ------------- %
$$
  \|\Psi\|_\varepsilon^2 := \sum_{n\in \mathbb{N}_0}\int_{\mathbb{R}}\left(|\psi_n'(x)|^2+(n+\varepsilon)|\psi_n|^2\right)\,\mathrm{d}x\,,
$$
 % ------------- %
where $\psi_n(x)$ are the coefficient functions of $\Psi$ in the decomposition \eqref{eq:separation}. As argued in \cite{EL18}, the norm $\|\Psi\|_\varepsilon = \sqrt{\|\Psi\|_\varepsilon^2}$ satisfies the parallelogram law, hence it is induced by an inner product, namely
 % ------------- %
$$
 (\Phi,\Psi)_\varepsilon := \sum_{n\in \mathbb{N}_0}\int_{\mathbb{R}}\left(\psi_n'(x) \bar\phi_n'(x)+(n+\varepsilon)\psi_n(x) \bar\phi_n(x)\right)\,\mathrm{d}x.
$$
 % ------------- %
Next we define the subspaces $\tilde D_j(\varepsilon) \subset D_j$, $j=1,2$ consisting of the functions  $\tilde \Psi$ which are projections of $\Psi$ with the components $P_n \tilde\psi^{(j)}_{\sqrt{n+\varepsilon}}$ and the coefficient sequence $\{P_n\}_{n=0}^\infty \in \ell^2(\mathbb{N}_0)$, where $\tilde\psi^{(j)}_{\sqrt{n+\varepsilon}}$ are the functions defined in \eqref{eq:psitilde}. In the further text, we drop the subscript/supperscript $j$, where it is not necessary. One can check easily that the functions $P_n\tilde \psi_{\sqrt{n+\varepsilon}}$ are normalized in such a way that $\|\tilde \Psi\|_\varepsilon = \|\{P_n\}_{n=0}^\infty\|_{\ell^2}$, and that
 % ------------- %
$$
  P_n := \int_{\mathbb{R}}\left(\psi_n'(x) \bar{\tilde\psi}'_{\sqrt{n+\varepsilon}}(x)+(n+\varepsilon)\psi_n(x) \bar{\tilde \psi}_{\sqrt{n+\varepsilon}}(x)\right)\,\mathrm{d}x\,.
$$
 % ------------- %
Integrating by parts, one finds
 % ------------- %
\begin{align}
  P_n &= \psi_n(0-)\bar{\tilde \psi}_{\sqrt{n+\varepsilon}}'(0-)- \psi_n(0+)\bar{\tilde \psi}_{\sqrt{n+\varepsilon}}'(0+)\nonumber
\\
   & = (\psi_n(0+) \bar K_+ + \psi_n(0-)\bar K_-)\frac{(n+\varepsilon)^{1/4}}{\sqrt{|K_+|^2+|K_-|^2}}\,.\label{eq:pn}
\end{align}
 % ------------- %
The expression on the left-hand side of condition \eqref{eq:condspec} can be rewritten as
 % ------------- %
\begin{equation}
 \mathbf{a}_{\alpha,\beta,\gamma}[\Psi] - \left(\frac{1}{2}-\varepsilon\right)\|\Psi\|^2_{L^2(\mathbb{R}^2)}  =   \|\Psi\|_\varepsilon^2 +\frac{1}{\beta}(\mathbf{b}_1[\Psi]+\mathbf{b}_2[\Psi]+\mathbf{b}_3[\Psi])\label{eq:psiepsineq}
\end{equation}
 % ------------- %
To deal with the second term on the right-hand side, we use formula \eqref{eq:pauli}. Let us assume that the functions $\Psi$, $\tilde\Psi_\varepsilon$ belong to a given $D_j$, $\tilde D_j(\varepsilon)$ $j=1,2\,$; we want to check that the mentioned expression will not change if we substitute $\tilde\Psi_\varepsilon = \{P_n \tilde \psi_{\sqrt{n+\varepsilon}}\}$ instead of $\Psi = \{\psi_n\}$. Using the values of $\tilde \psi_{\sqrt{n+\varepsilon}}(0+)$, $\tilde \psi_{\sqrt{n+\varepsilon}}(0-)$ from~\eqref{eq:psitilde} and $P_n$ from \eqref{eq:pn} we have
 % ------------- %
\begin{align*}
  P_n \begin{pmatrix}\tilde \psi_{\sqrt{n+\varepsilon}}(0+)\\\tilde \psi_{\sqrt{n+\varepsilon}}(0-)\end{pmatrix} & = (\psi_n(0+)\bar K_+ + \psi_n(0-)\bar K_- ) \frac{(n+\varepsilon)^{1/4}}{\sqrt{|K_+|^2+|K_-|^2}} \begin{pmatrix}K_+\\ K_-\end{pmatrix}
\\
& \quad\times \frac{1}{(n+\varepsilon)^{1/4}\sqrt{|K_+|^2+|K_-|^2}}
\\[.3em]
 & = \frac{1}{|K_+|^2+|K_-|^2}\begin{pmatrix}K_+ & 0 \\ 0 & K_-\end{pmatrix} \cdot \begin{pmatrix}1 & 1 \\ 1 & 1\end{pmatrix} \cdot \begin{pmatrix}\bar K_+ & 0 \\ 0 & \bar K_-\end{pmatrix} \cdot \begin{pmatrix}\psi_n(0+)\\ \psi_n(0-)\end{pmatrix}\,,
\end{align*}
 % ------------- %
where `$\cdot$' denotes matrix multiplication. We have used
 % ------------- %
\begin{align*}
  (\psi_n(0+)\bar K_+ + \psi_n(0-)\bar K_- ) &= \begin{pmatrix}1 & 1\end{pmatrix}\cdot \begin{pmatrix}\bar K_+ & 0 \\ 0 & \bar K_-\end{pmatrix}\cdot \begin{pmatrix}\psi_n(0+)\\ \psi_n(0-)\end{pmatrix}\,,\\
 \begin{pmatrix}K_+\\ K_-\end{pmatrix} &= \begin{pmatrix} K_+ & 0 \\ 0 & K_-\end{pmatrix}\begin{pmatrix}1\\1\end{pmatrix}\,.
\end{align*}
 % ------------- %
Both the coefficient functions $\psi_n$ and $\psi_{n-1}$ belong to $D_j$, hence we can write
 % ------------- %
$$
  \begin{pmatrix}\psi_n(0+)\\\psi_n(0-)\end{pmatrix} = Q_n \begin{pmatrix}K_+\\ K_-\end{pmatrix}\,,
$$
 % ------------- %
where $Q_n$ is a normalization constant, the value of which is not important. Defining
 % ------------- %
$$
  M := \frac{1}{|K_+|^2+|K_-|^2}\begin{pmatrix}K_+ & 0 \\ 0 & K_-\end{pmatrix}  \begin{pmatrix}1 & 1 \\ 1 & 1\end{pmatrix}  \begin{pmatrix}\bar K_+ & 0 \\ 0 & \bar K_-\end{pmatrix}
$$
 % ------------- %
and writing $M^\dagger$ for the Hermitian conjugate of this matrix, the equation we have to check in order to justify the claim we made about the second term in \eqref{eq:psiepsineq} simplifies to
 % ------------- %
$$
  |Q_n|^2 \begin{pmatrix}K_+\\ K_-\end{pmatrix}^\dagger M^\dagger \Sigma M \begin{pmatrix}K_+\\ K_-\end{pmatrix} = |Q_n|^2 \begin{pmatrix}K_+\\ K_-\end{pmatrix}^\dagger \Sigma \begin{pmatrix}K_+\\ K_-\end{pmatrix}
$$
 % ------------- %
which is not difficult to verify, for instance, using \emph{Mathematica}.

In the rest of the proof, one can proceed in a way similar to that of \cite{EL18,Sol04b}. If we replace $\Psi$ by $\tilde \Psi_\varepsilon$ in \eqref{eq:psiepsineq}, the first term does not increase, because $\tilde \Psi_\varepsilon$ is a projection of the function $\Psi$. The second term is not affected by the change, as argued previously, thus the inequality \eqref{eq:condspec} holds true for $\tilde \Psi_\varepsilon$ as well.

Consider two subspaces $\mathcal{F}, \mathcal{F}' \in \mathfrak{F}$ such that $\mathcal{F}\subset \mathcal{F}'$ and $\mathcal{F}\in \tilde D_j(\varepsilon)$. Should there be an element $\Psi\in D_j$ orthogonal to $\tilde D_j(\varepsilon)$, then we would have $P_n = 0$ for all $n\in \mathbb{N}$ which in turn implies $\mathbf{b}_1[\Psi]+\mathbf{b}_2[\Psi]+\mathbf{b}_3[\Psi] = 0$. In such a case, however, the inequality \eqref{eq:condspec} would not hold, which is a contradiction. Consequently,
 % ------------- %
$$
  N_-\left(\frac{1}{2}-\varepsilon, \left.\mathbf{H}_{\alpha,\beta,\gamma}\right|_{D_j}\right) = \mathrm{max}_{\mathcal{F}\in\mathfrak{F}_j(\varepsilon),\, \mathcal{F}\in \tilde D_j(\varepsilon)}\: \mathrm{dim}\mathcal{F}\,.
$$
 % ------------- %
For each $\tilde \Psi_\varepsilon \in \tilde D(\varepsilon)$ we have
 % ------------- %
\begin{align*}
\|\tilde \Psi_\varepsilon\|_\varepsilon^2 +& \frac{1}{\beta} (\mathbf{b}_1[\Psi]+\mathbf{b}_2[\Psi]+\mathbf{b}_3[\Psi]) \\[.3em] = &\sum_{n\in \mathbb{N}_0} |P_n|^2
+\sum_{n\in \mathbb{N}}\frac{\sqrt{n}}{2\sqrt{2}}\:\mathrm{Re}\left[\begin{pmatrix}\psi_n(0+)\\\psi_n(0-)\end{pmatrix}^\dagger \Sigma \begin{pmatrix}\psi_{n-1}(0+)\\\psi_{n-1}(0-)\end{pmatrix}\right]
\end{align*}
 % ------------- %
Using next the relation
 % ------------- %
$$
  P_n = \begin{pmatrix}\bar K_+ & \bar K_-\end{pmatrix}\begin{pmatrix}\psi_n(0+)\\ \psi_n(0-)\end{pmatrix}\frac{(n+\varepsilon)^{1/4}}{\sqrt{|K_+|^2+|K_-|^2}}
$$
 % ------------- %
which follows from \eqref{eq:pn}, the fact that the matrix $\Sigma$ can be diagonalized using the unitary matrix composed of its eigenvectors with the eigenvalues $2\sqrt{2}\beta \mu_j^{-1}$, and the fact that we have restricted ourselves to the subspace $\tilde D_j (\varepsilon)$ corresponding to one of its eigenspaces, we can rewrite the above expression as
 % ------------- %
\begin{align*}
  \|\tilde \Psi_\varepsilon\|_\varepsilon^2 + \frac{1}{\beta} (\mathbf{b}_1[\Psi]+\mathbf{b}_2[\Psi]+\mathbf{b}_3[\Psi]) & = \sum_{n\in \mathbb{N}_0} |P_n|^2 + 2\mu_j^{-1}\sum_{n\in\mathbb{N}}j_{n,n-1}(\varepsilon)\,\mathrm{Re\,}(\bar P_n P_{n-1})
  \\
  & = \|g\|^2_{\ell^2(\mathbb{N}_0)}+\mu_j^{-1}(\mathbf{J}(\varepsilon)g,g)_{\ell^2(\mathbb{N}_0)}
 \\
 & = (I+\mu_j^{-1}\mathbf{J}(\varepsilon)g,g),
\end{align*}
 % ------------- %
where $g = \{P_n\}\in\ell^2(\mathbb{N}_0)$. The claim of the theorem now follows from the decomposition of the operator $\mathbf{H}_{\alpha, \beta, \gamma}$ to the subspaces $D_j$ in combination with the spectrum of $\mathbf{J}(\varepsilon)$.
\end{proof}

Next we proceed in analogy with \cite[Theorem 3.2]{Sol04b} and \cite[Theorem 11]{EL18}. One cannot apply the previous theorem directly to the case $\varepsilon=0$, because $j_{1,0} = \infty$. Restricting the quadratic form to the subspace of $g = {P_n}$ with $P_0 = 0$, however, the limit $\varepsilon \to 0$ leads to the operator $\mathbf{J}_0$ introduced in the opening of this section. The price we pay for this restriction is the possible change of the number of eigenvalues by one. This leads to the first main result about the discrete spectrum.
 % ------------- %
\begin{theorem}\label{thm:discrete2}
Let $\beta\ne 0$, $\,\mu_1>1$ and $\mu_2\leq 1$. Then \mbox{$\big|N_-(\frac{1}{2},\mathbf{H}_{\alpha,\beta,\gamma}) - N_+(\mu_1,\mathbf{J}_0)\big|\le 1$;} the analogous relation holds if the roles of $\mu_1,\,\mu_2$ are interchanged. If
both $\mu_1, \mu_2 >1$, we have $\big|N_-(\frac{1}{2},\mathbf{H}_{\alpha,\beta,\gamma}) - N_+(\mu_1,\mathbf{J}_0) - N_+(\mu_2,\mathbf{J}_0)\big|\le 2$. On the other hand, for $\beta=0$ and $\frac{4+|\gamma|^2}{2\sqrt{2}\alpha}>1$ the inequality $\big|N_-(\frac{1}{2},\mathbf{H}_{\alpha, 0,\gamma}) - N_+\big(\frac{4+|\gamma|^2}{2\sqrt{2}\alpha},\mathbf{J}_0\big)\big|\le 1$ holds.
\end{theorem}
 % ------------- %

Finally, let us look how the cardinality of the discrete spectrum behaves when the coupling parameters approach the critical values.
 % ------------- %
\begin{theorem}\label{thm:discrete3}
We have the following asymptotic formul{\ae}:
 % ------------- %
\begin{enumerate}
\item[(i)]
If $0< \beta \leq 2 \sqrt{2}$ and $\alpha \to \frac{1}{\beta}\left[\frac{|\gamma|^2}{2\sqrt{2}(2\sqrt{2}-\beta)}-(|\gamma|^2-4)-\sqrt{2}(2\sqrt{2}-\beta)\right]$ as the left limit, then $\mu_1 \to 1+$ and $N_-(\frac{1}{2},\mathbf{H}_{\alpha, \beta,\gamma})\sim \frac{1}{4\sqrt{2}\sqrt{\mu_1-1}}$.
 % ------------- %
\item[(ii)]
If $ \beta \geq 2 \sqrt{2}$ and $\alpha \to \frac{1}{\beta}\left[\frac{|\gamma|^2}{2\sqrt{2}(2\sqrt{2}-\beta)}-(|\gamma|^2-4)-\sqrt{2}(2\sqrt{2}-\beta)\right]$ as the left limit, then $\mu_2 \to 1+$ and $N_-(\frac{1}{2},\mathbf{H}_{\alpha, \beta,\gamma})\sim \frac{1}{4\sqrt{2}\sqrt{\mu_2-1}}$.
 % ------------- %
\item[(iii)]
If $\gamma = 0$ and $\alpha \beta =4$, $\beta\geq 2\sqrt{2}$, then $\mu_1 = \mu_2$ and $N_-(\frac{1}{2},\mathbf{H}_{\alpha, \beta,\gamma})\sim \frac{1}{2\sqrt{2}\sqrt{\frac{\beta}{2\sqrt{2}}-1}}$ in the limit $\beta \to 2\sqrt{2}+$.
 % ------------- %
\item[(iv)] If $\beta = 0$ then $N_-(\frac{1}{2},\mathbf{H}_{\alpha, 0,\gamma})\sim \frac{2^{1/4}}{4}\sqrt{\frac{\alpha}{4+|\gamma|^2-2\sqrt{2}\alpha}}$ holds as $\frac{4+|\gamma|^2}{2\sqrt{2}\alpha}\to 1+$.
\end{enumerate}
 % ------------- %
\end{theorem}
 % ------------- %
\begin{proof}
One can check directly that for a nonzero $\beta$ satisfying the assumption we have
 % ------------- %
$$
\mu_{1,2}^{-1} = \frac{\alpha \beta+|\gamma|^2-4}{2\sqrt{2}\beta}\mp \frac{\sqrt{(\alpha \beta+|\gamma|^2-4)^2+16|\gamma|^2}}{2\sqrt{2}\beta}+\frac{2\sqrt{2}}{\beta} \to 1-
$$
 % ------------- %
in the limit
 % ------------- %
$$
\alpha \to \frac{1}{\beta}\left[\frac{|\gamma|^2}{2\sqrt{2}(2\sqrt{2}-\beta)}-(|\gamma|^2-4)-\sqrt{2}(2\sqrt{2}-\beta)\right]-
$$
 % ------------- %
and to proceed as in \cite[Thm~13]{EL18}  and \cite[eq. (3.10)]{Sol04b}. Specifically, by \cite[Thm~3.3]{Sol04b}, the Jacobi operator $\mathbf{J}_0$ spectrum satisfies  $N_+(\mu, \mathbf{J}_0) \sim \frac{1}{4\sqrt{2}\sqrt{\mu-1}}$ as $\mu \to 1+$; by Theorem~\ref{thm:discrete2} we then obtain the claims (i) and (ii). The relation $\mu_1=\mu_2$ in claim (iii) can be easily verified under the condition $\gamma = 0$ and $\alpha \beta =4$, the rest follows from Theorem~\ref{thm:discrete2}. Finally, claim (iv) follows from the last part of Theorem~\ref{thm:discrete2}.
\end{proof}

\bigskip

%%%%%%%%%%%%%%%%%%%%%%%%%%%
\subsection*{Acknowledgements}
%%%%%%%%%%%%%%%%%%%%%%%%%%%

P.E. was supported by the Czech Science Foundation within the project 21-07129S, J.L. by Czech Science Foundation within the project 22-18739S. The authors thank the reviewer for useful remarks which improved the paper.

%%%%%%%%%%%%%%%%%%%%%%%%%%%
\appendix
\section{Appendix: Asymptotics of solutions of the Jacobi equation} \label{sec:appa}
%%%%%%%%%%%%%%%%%%%%%%%%%%%
In this appendix, we give technical results on large $n$ asymptotics of the solutions of the equation~\eqref{eq:jacobi}. First, we give without the proof the result by Birkhoff and Adams~\cite[Theorem 8.36]{Ela99} on the asymptotics of a more general system. A misprint was corrected according to \cite{NS06}.

\begin{theorem}\label{thm:BirAd}
Let the Jacobi equation
\begin{equation}
  C(n+1)+p_1(n)C(n)+p_2(n)C(n-1) = 0 \label{eq:jacgen}
\end{equation}
have the following large $n$ expansion of the functions $p_1$ and $p_2$
$$
  p_1(n) = \sum_{j=0}^\infty a_j n^{-j}, \quad p_2(n) = \sum_{j=0}^\infty b_j n^{-j}, \quad b_0 \ne 0\,.
$$
Let $\lambda_{\pm}$ be the roots of the quadratic equation $\lambda^2+a_0\lambda+b_0=0$.
Then the system~\eqref{eq:jacgen} has two linearly independent solutions $C^{\pm}(n)$ with the following large $n$ asymptotics.
\begin{enumerate}
\item For $\lambda_+ \ne \lambda_-$ it holds
$$
  C^{\pm}(n) \sim \lambda_{\pm}^n n ^{d_{\pm}}\,,\quad d_{\pm} = \frac{a_1\lambda_{\pm}+b_1}{a_0\lambda_{\pm}+2b_0}\,.
$$
\item For $\lambda_+ = \lambda_- = \lambda$ it holds
$$
  C^{\pm}(n) \sim \lambda^n \mathrm{e}^{\pm \delta \sqrt{n}} n^\kappa\,,\quad \delta =2\sqrt{\frac{a_0a_1-2b_1}{2b_0}}\,,\quad \kappa = \frac{1}{4}+\frac{b_1}{2b_0}\,.
$$
\end{enumerate}
\end{theorem}

We use the previous theorem for finding the asymptotics of the solutions of the equation~\eqref{eq:jacobi}. We slightly generalize \cite[Lemma 3.3]{NS06}, allowing $\mu$ to be negative.

\begin{lemma}\label{lem:C:asymptotics}
Let $\mu\in \mathbb{R}$ and $\Lambda \in \mathbb{C}\backslash (\frac{1}{2},\infty)$. Then the system \eqref{eq:jacobi} has two linearly independent solutions $C^{\pm}_n$ with the following large $n$ asymptotics
\begin{enumerate}
\item For $\mu^2<1$ it holds $C^{\pm}_n\sim (\mu + i \sqrt{1-\mu^2})^{\pm n}\, n^{-\frac{1}{2}\mp \frac{i \Lambda \mu}{2\sqrt{1-\mu^2}}}$.
\item For $\mu^2>1$ it holds $C^{\pm}_n\sim (-\mu +\sqrt{\mu^2-1})^{\pm n}\, n^{-\frac{1}{2}\pm \frac{\Lambda \mu}{2\sqrt{\mu^2-1}}}$.
\item For $\mu=1$ it holds $C^{\pm}_n\sim (-1)^n\mathrm{e}^{\pm 2\sqrt{-\Lambda n}}\, n^{-1/4}$.
\item For $\mu=-1$ it holds $C^{\pm}_n\sim \mathrm{e}^{\pm 2\sqrt{-\Lambda n}}\, n^{-1/4}$.
\end{enumerate}
\end{lemma}
\begin{proof}
The results follows directly from Theorem~\ref{thm:BirAd}. The system~\eqref{eq:jacobi} is a particular case of the system~\eqref{eq:jacgen} with
$$
  a_0 = 2\mu\,, \quad a_1 = -\mu(1+\Lambda)\,,\quad b_0 =1, \quad b_1 = -1\,.
$$
Using these values we obtain the quadratic equation $\lambda^2+2\mu \lambda +1 = 0$ with the roots $\lambda_{\pm} = -\mu \pm \sqrt{\mu^2-1}$. In the first two cases, we use the first case of Theorem~\ref{thm:BirAd} with $d_{\pm} = -\frac{1}{2}-\frac{\Lambda}{2}\frac{\mu(-\mu \pm \sqrt{\mu^2-1})}{\mu(-\mu \pm \sqrt{\mu^2-1})+1}$. In the latter two cases we apply the second case of Theorem~\ref{thm:BirAd} since $\lambda_{\pm} = -1$ for $\mu =1$ and $\lambda_{\pm} = 1$ for $\mu =-1$.
\end{proof}

The following lemma is taken from \cite[eq. (3.20)]{NS06}.

\begin{lemma}\label{lem:Cnidentity}
The solutions $C_n$ of the equation
\begin{equation}
  Q_{n+1}C_{n+1}+ P_n C_n+Q_n C_{n-1}=0\,,\quad n\in \mathbb{N}_0\label{eq:Cnidentity}
\end{equation}
with $Q_n\in\mathbb{R}$, $Q_0 = 0$ satisfy the identity
$$
  \sum_{n=0}^N |C_n|^2\,\mathrm{Im\,}P_n = -Q_{N+1}\,\mathrm{Im}(C_{N+1}\overline{C_{N}})\,.
$$
\end{lemma}
\begin{proof}
The identity follows from multiplying~\eqref{eq:Cnidentity} by $\overline{C_n}$, taking the imaginary part and taking the sum of $n$ from 0 to $N$.
\end{proof}

Its corollary is also inspired by \cite{NS06}, however, its form is slightly different than in the mentioned paper.

\begin{corollary}\label{cor:Cnidentity}
Solutions of the equation~\eqref{eq:jacobi} satisfy
\begin{equation}
 2\mu \sum_{n=0}^N |C_n|^2 (n+1/2)^{1/2}\,\mathrm{Im\,}\zeta_{n}(\Lambda) = - (N+1)^{1/2}\left((N+1)^2-\frac{1}{4}\right)^{1/4}\,\mathrm{Im\,}(C_{N+1}\overline{C_N})\,.\label{eq:Cnidentity2}
\end{equation}
\end{corollary}
\begin{proof}
Taking in the Lemma~\ref{lem:Cnidentity} the values $Q_n = n^{1/2} (n^2-1/4)^{1/4}$ and $P_n = 2\mu (n+1/2)^{1/2}\zeta_n(\Lambda)$, which correspond to~\eqref{eq:jacobi} multiplied by $(n+1/2)^{1/4}$, the result follows.
\end{proof}

\end{document}